\documentclass[11pt]{article}
\usepackage{color}
\usepackage{epsfig}
\usepackage{subfig}
\usepackage{mathrsfs}
\usepackage{booktabs}
\usepackage{graphicx}
\usepackage{epstopdf}
\usepackage{multirow}
\usepackage{amsmath,amssymb}
\usepackage{amsthm}
\usepackage{fancyhdr}
\usepackage{tabularx}
\usepackage{xfrac}
\usepackage{CJK}
\usepackage{mathptmx}
\usepackage{mathrsfs}
\usepackage{graphics}
\usepackage{times}

\newtheorem{prop}{Proposition}[section]

\newtheorem{remark}{Remark}[section]

\begin{document}

\title{Contingency Identification of Cascading Failures in Power Transmission Networks}

\author{Chao Zhai, Hehong Zhang, Gaoxi Xiao and Tso-Chien Pan \thanks{Chao Zhai, Hehong Zhang, Gaoxi Xiao and Tso-Chien Pan are with Institute of Catastrophe Risk Management, Nanyang Technological University, 50 Nanyang Avenue, Singapore 639798. They are also with Future Resilient Systems, Singapore-ETH Centre, 1 Create Way, CREATE Tower, Singapore 138602. Chao Zhai, Hehong Zhang and Gaoxi Xiao are also with School of Electrical and Electronic Engineering, Nanyang Technological University. Corresponding author: Gaoxi Xiao. Email: egxxiao@ntu.edu.sg}}

\maketitle

\begin{abstract}
Due to the evolving nature of power systems and the complicated coupling relationship of power devices, it has been a great challenge to identify the contingencies that could trigger cascading blackouts of power systems. This paper provides an effective approach to identifying the initial contingency in power transmission networks, which are equipped with flexible alternating current transmission system (FACTS) devices, high-voltage direct current (HVDC) links and protective relays. Essentially, the problem of contingency identification is formulated in the framework of nonlinear programming, which can be solved by the Jacobian-Free Newton-Krylov (JFNK) method to circumvent Jacobian matrix and reduce the computational cost. Notably, the proposed identification approach is also applied to complicated cascading failure models of power systems. Finally, numerical simulations are carried out to validate the proposed identification approach on IEEE $118$ Bus Systems. The proposed approach succeeds in reconciling the rigorous optimization formulation with the practical modeling of cascading blackouts.
\end{abstract}


\section{Introduction}
The past decades have witnessed several large blackouts in the world such as India Blackout (2012), US-Canada Blackout (2003), Italy Blackout (2003) and Southern Brazil Blackout (1999) to name just a few, which have left millions of residents without power supply and caused huge financial losses \cite{mcl09}. In such catastrophe events, the initial contingencies ($e.g.$ extreme weather, terrorist attack and operator error) play a crucial role in triggering the cascading outage of power systems. It is reported that the mal-operation of a protection relay is the key ``trigger" of the final line outage sequence in most blackouts \cite{beck05}. For instance, conventional relays may lead to unselective tripping under high load conditions, which could initiate the chain reaction of branch outages under certain conditions (e.g., a wrong relay operation of Sammis-Star line in the 2003 US-Canada Blackout \cite{beck05}). The reliability and resilience of power grids are closely related to the proactive elimination of disruptive initial contingencies. Thus, it is vital to identify the initial contingency that causes the most severe blackouts and work out remedial schemes against cascading blackouts in advance.

In practice, electrical power devices such as FACTS devices, HVDC links and protective relays serve as the major protective barrier against cascading blackouts. To be specific, FACTS devices significantly contribute to the stability improvement of power systems, while HVDC links behave like a ``firewall" to prevent the propagation of cascading outages. Actually, the FACTS devices have been widely installed in power transmission networks to improve the capability of power transmission, controllability of power flow, damping of power oscillation and post-contingency stability. As a series FACTS device, the thyristor-controlled series capacitor (TCSC) allows fast and continuous adjustments of branch impedance in order to control the power flow and improve the transient stability \cite{jov05}. In addition, the HVDC links assist in preventing cascades propagation and restoring the power flow after faults. For example, Qu\'{e}bec power system in Canada survived the cascades in the 2003 US-Canada Blackout due to its DC interconnection to the US power systems \cite{beck05}. As the most common protection device, protective relays of power system react passively to the system oscillation and promptly remove the overloading elements without affecting the normal operation of the rest of the system. Meanwhile it allows for time delay of abnormal oscillations to neglect the trivial disturbances and avoid the overreaction to the transient state changes \cite{jia16}. It is necessary to take into account the protection mechanism of the above power devices for the practical cascading dynamics of power systems.

Owing to simplicity, efficiency and scalability in the simulation, the DC power flow model has been widely adopted to investigate cascading failures of power systems \cite{alm15,yan15}. It is demonstrated that the DC power flow model is able to assess the vulnerability of power grids and reveal informative details of cascading failure process, including the size, contributing factors and the duration of cascading failures \cite{yan15}. Additionally, the model predictive control can be applied to mitigate the cascading effect of severe line-overload disturbances in power systems \cite{alm15}. Actually, the DC power flow model is usually regarded as a good substitute for the AC based model in high voltage power transmission networks \cite{zhai17a,stot09}. As a result, the DC power flow equation is employed in this work to compute the transmission power on branches of power transmission networks.

So far, cascading blackouts of power systems have been investigated through two distinct routes. Specifically, some researchers aim at the strict mathematical formulation for the exploration of vulnerable elements in power systems regardless of the transient response and protection mechanisms \cite{alm15,tae16}, while others focus on the practical physical process and accurate modeling of cascading blackouts \cite{jia16,yan15}. While the former may fail to reflect the real physical characteristic of cascading failures, the latter is in lack of a rigorous theoretical framework. This work attempts to fill the gap between the practical modeling of cascading blackouts and the strict mathematical formulation by properly decoupling the optimization problem and cascading dynamics of power grids. The main contributions of this paper are listed as follows
\begin{enumerate}
  \item Propose the cascading dynamics of power transmission networks equipped with FACTS devices, HVDC links and protective relays.
  \item Formulate the problem of contingency identification with nonlinear programming and solve it via the efficient numerical method.
  \item Validate the proposed approach on the large-scale power transmission networks using different protection schemes.
\end{enumerate}

The outline of this paper is organized as follows. Section \ref{sec:prob} presents the cascading dynamics of power transmission networks. Section \ref{sec:opt} provides the optimization formulation and theoretical results on the contingency identification, followed by numerical methods in Section \ref{sec:num}. Next, the identification approach is validated in Section \ref{sec:sim}. Finally, we conclude the paper and discuss future work in Section \ref{sec:con}.

\section{Cascading Dynamics}\label{sec:prob}
This section aims to characterize the cascading evolution of power transmission networks subject to the initial contingency and system stresses. Figure \ref{cascade} presents the cascading process of power systems after the initial contingency is added on the system. First of all, the FACTS devices take effect to adjust the branch admittance and balance the power flow for relieving the stress of power networks. If the stress is not eliminated, protective relays will be activated to serve the overloading branches on the condition that the timer of circuit breakers runs out of the preset time. Under certain circumstances, the outage of overloading branches may result in severer stress of power networks and ends up with having cascading blackouts of power systems. To describe the cascading process, we introduce the concept of cascading step. Essentially, a cascading step is defined as one topological change ($e.g.$, one branch outage) of power networks due to contingencies, human interferences or the branch overloads. The models of the DC power flow, FACTS devices, HVDC links and protective relays are presented in sequence.

\begin{figure}
\scalebox{0.07}[0.07]{\includegraphics{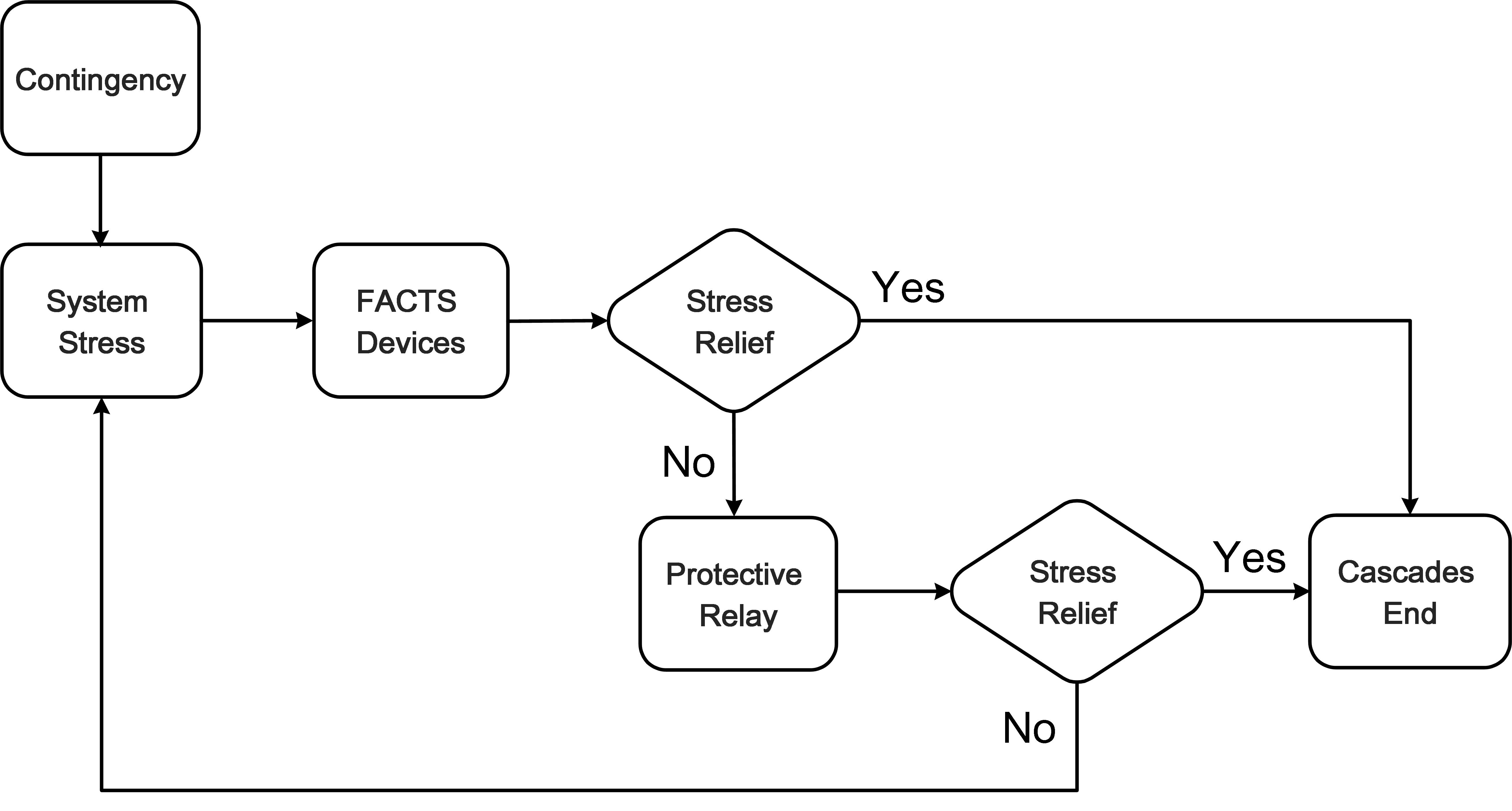}}\centering
\caption{\label{cascade} Cascading failure process of power transmission networks.}
\end{figure}

\subsection{DC Power Flow Model}
For high-voltage power transmission networks, the DC power flow equation is well qualified to describe the quantitative relationship of injected bus power, branch susceptance and voltage angle as follows
\begin{equation}\label{dc_pfe}
    P_b=A^Tdiag(B)A\theta
\end{equation}
where $A$ denotes the branch-bus incidence matrix \cite{stag68} and $\theta$ refers to the vector of voltage angles. $P_b$ represents the vector of injected power on each bus. Additionally, $B=(B_{1},B_{2},...,B_{n})$ is the susceptance vector for branches, and each element $B_{i}$ is given by
$$
B_{i}=-\frac{1}{X_{C,i}+X_i}, \quad i\in I_n=\{1,2,...,n\}
$$
where $X_{C,i}$ denotes the reactance of TCSC equipped on Branch $i$, and $X_i$ represents the original reactance of Branch $i$. The DC power flow equation (\ref{dc_pfe}) can be solved as
$$
\theta=\left(A^Tdiag(B)A\right)^{-1^*}P_b
$$
where the operator $-1^{*}$ denotes the operation of matrix inverse, which is defined in \cite{cz17}. Then the vector of transmission power on each branch can be computed by
\begin{equation}\label{p_e}
    P_e=diag(B)A\left(A^Tdiag(B)A\right)^{-1^*}P_b
\end{equation}
Notably, the generator bus connected to the largest generating station is selected as the slack bus,
and thus the power variation of slack bus accounts for a small percentage of its generating capacity.

\subsection{FACTS Devices}
\begin{figure}
\scalebox{0.07}[0.07]{\includegraphics{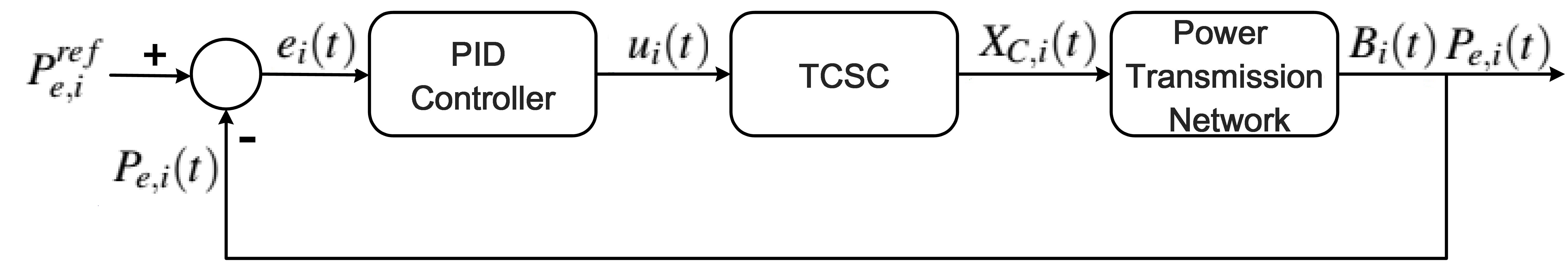}}\centering
\caption{\label{tcsc_con} Control diagram of TCSC on branches.}
\end{figure}
FACTS devices can greatly enhance the stability and transmission capability of power systems. As an effective FACTS device, TCSC has been widely installed to control the branch impedance and relieve system stresses. The dynamics of TCSC is described by a first order dynamical model \cite{pas95}
\begin{equation}\label{tcsc}
    T_{C,i}\frac{d{X}_{C,i}}{dt}=-X_{C,i}+X_{ref,i}+u_i, \quad X_{\min,i}\leq X_{C,i}\leq X_{\max,i}\quad i\in I_n
\end{equation}
where $X_{ref,i}$ refers to its reference reactance of Branch $i$ for the steady power flow. $X_{\min,i}$
and $X_{\max,i}$ are the lower and upper bounds of the branch reactance $X_{C,i}$ respectively and $u_i$ represents the supplementary control input, which is designed to stabilize the disturbed power system \cite{son00}. For simplicity, PID controller is adopted to regulate the power flow on each branch
\begin{equation}\label{pid}
    u_i(t)=K_P\cdot e_i(t)+K_I\cdot\int_{0}^{t}e_i(\tau)d\tau+K_D\cdot\frac{de_i(t)}{dt}
\end{equation}
where $K_P$, $K_I$ and $K_D$ are tunable coefficients, and the error $e_i(t)$ is given by
$$
e_i(t)=\left\{
         \begin{array}{ll}
           P^{ref}_{e,i}-|P_{e,i}(t)|, & \hbox{$|P_{e,i}(t)|\geq P^{ref}_{e,i}$;} \\
           0, & \hbox{otherwise.}
         \end{array}
       \right.
$$
Here, $P^{ref}_{e,i}$ and $P_{e,i}(t)$ denote the reference transmission power and the actual transmission power of Branch $i$, respectively. Note that TCSC fails to function when the transmission line is severed. Figure \ref{tcsc_con} presents the diagram about the operation of TCSC via PID controller to reach the reference transmission power. First of all, we compute the error $e_i(t)$ between the actual power $P_{e,i}(t)$ and the reference power $P^{ref}_{e,i}$. Next, the PID controller produces the control input $u_i(t)$ based on $e_i(t)$, which regulates the reactance of TCSC on Branch $i$. Finally, the actual power $P_{e,i}(t)$ will converge to the reference power $P^{ref}_{e,i}$ as time goes.

\subsection{HVDC Links}
\begin{figure}
\scalebox{0.075}[0.075]{\includegraphics{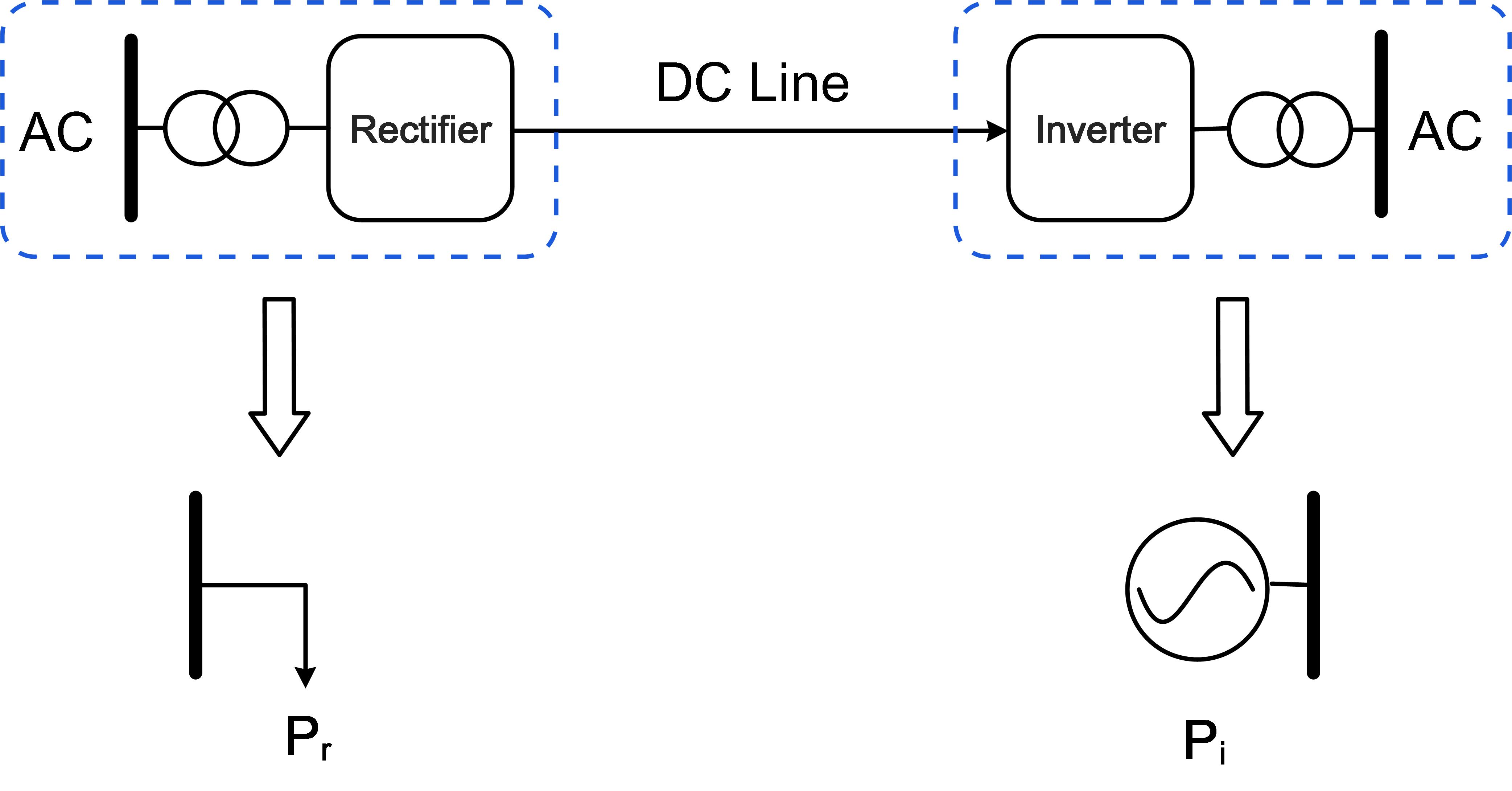}}\centering
\caption{\label{hvdc}Schematic diagram of monopolar HVDC link and its equivalent circuit.}
\end{figure}
In practice, the HVDC link works as a protective barrier to prevent the propagation of cascading outages,
and it is normally composed of a transformer, a rectifier, a DC line and an inverter (see Fig.~\ref{hvdc}). Actually, the rectifier terminal can be regarded as a bus with real power consumption $P_{r}$, and the inverter terminal can be treated as a bus with real power generation $P_{i}$. The direct current from the rectifier to the inverter is computed as follows \cite{kun94}
$$
I_d=\frac{3\sqrt{3}(\cos\alpha-\cos\gamma)}{\pi(R_{cr}+R_L-R_{ci})},
$$
where $\alpha\in[\pi/30,\pi/2]$ denotes the ignition delay angle of the rectifier, and $\gamma\in[\pi/12,\pi/9]$ represents the extinction advance angle of the inverter. $R_{cr}$ and $R_{ci}$ refer to the equivalent communicating resistances for the rectifier and inverter, respectively. Additionally, $R_L$ denotes the resistance of the DC transmission line. Thus the power consumption at the rectifier terminal is
\begin{equation}\label{pr}
    P_r=\frac{3\sqrt{3}}{\pi}I_d\cos\alpha-R_{cr}I^2_d,
\end{equation}
and at the inverter terminal is
\begin{equation}\label{pi}
    P_i=\frac{3\sqrt{3}}{\pi}I_d\cos\gamma-R_{ci}I^2_d=P_r-R_LI^2_d.
\end{equation}
Notably, $P_r$ and $P_i$ keep unchanged when $\alpha$ and $\gamma$ are fixed.

\subsection{Protective Relay}
The protective relays are indispensable components in power systems protection and control. When the transmission power exceeds the given threshold of the branch, the timer of circuit breaker starts to count down from the preset time \cite{jia16}. Once the timer runs out of the preset time, the transmission line is severed by circuit breakers and its branch admittance becomes zero. Specifically, the vector of branch susceptance at the $k$-th cascading step is given by
\begin{equation}\label{relay}
    B^k=G(P^{k-1}_{e},\sigma)\circ F(B^{k-1})
\end{equation}
where the operator $\circ$ denotes the Hadamard product, and $\sigma=(\sigma_1,\sigma_2,...,\sigma_n)$ represents the threshold vector of transmission power on each branch. $F(B^{k-1})$ provides the vector of branch susceptance at the $(k-1)$-th cascading step, and it updates constantly due to the dynamics of FACTS devices. Additionally, the vector function $G(P^{k-1}_{e},\sigma)$ is used to characterize the branch outage as follows
$$
G(P^{k-1}_{e},\sigma)=\left(
                        \begin{array}{c}
                          g(P^{k-1}_{e,1},\sigma_1) \\
                          g(P^{k-1}_{e,2},\sigma_2) \\
                          . \\
                          g(P^{k-1}_{e,n},\sigma_n) \\
                        \end{array}
                      \right)\in R^n
$$
And each element of $G(P^{k-1}_{e},\sigma)$ is a step function as follows
$$
g(P^{k-1}_{e,i},\sigma_i)=\left\{
                            \begin{array}{ll}
                              0, & \hbox{$|P^{k-1}_{e,i}|>\sigma_i$ and $t_c>T$;} \\
                              1, & \hbox{otherwise.}
                            \end{array}
                          \right.
$$
where $T$ is the preset time of the timer in protective relays, and $t_c$ denotes the counting time of the timer. Intuitively, the branch outage occurs when its transmission power is larger than the threshold and meanwhile its timer runs out.

The evolution time of cascading failure is introduced to allow for the time factor of cascading blackouts.
Essentially, the time interval between two consecutive cascading steps basically depends on the preset time of the timer in protective relays \cite{jia16}. Thus, the evolution time of cascading failure is roughly estimated by $t=kT$ at the $k$-th cascading step.

\section{Optimization Formulation}\label{sec:opt}

Since cascading blackouts result in the severe damage of power transmission, we focus on the power transmission at the end of cascading outages and thus design the cost function as follows
\begin{equation}\label{cost}
J(\delta, B^m)=\frac{1}{2}\|P^m_{e}(\delta)\|^2
\end{equation}
where $P^m_{e}(\delta)$ denotes the vector of transmission power on each branch at the $m$-th cascading step, and $\delta\in[\underline{\delta},\bar{\delta}]$ characterizes the admittance change of the selected branch caused by the initial contingency. Specifically, the vector of transmission power $P^m_{e}(\delta)$ after the contingency can be computed by
\begin{equation*}
    P^m_{e}(\delta)=diag(B^m)A(A^T diag(B^m)A)^{-1^*}P_b
\end{equation*}
with
$$
B^k=G(P^{k-1}_{e},\sigma)\circ F(B^{k-1}), \quad k\in I_m=\{1,2,...,m\}.
$$
As mentioned before, $F(B^{k-1})$ characterizes the dynamical adjustment of FACTS devices at the $k$-th cascading step with $B^{1}=B^{0}+\delta$. Notably, $P_b$ refers to the vector of injected power on buses after the rectifier and inverter terminals of HVDC links are treated as the loads and generators, respectively. Therefore, the problem of identifying initial contingencies in power transmission networks is formulated as
\begin{equation}\label{formulation}
\begin{split}
&~~~~~~\min_{\delta} J(\delta,B^m) \\
&s.~t.~\underline{\delta}\leq\delta\leq\bar{\delta} \\
&~~~~~~B^k=G(P^{k-1}_{e},\sigma)\circ F(B^{k-1}),~k\in I_m \\
&~~~~~~P^k_{e}(\delta)=diag(B^k)A(A^T diag(B^k)A)^{-1^*}P_b \\
\end{split}
\end{equation}
where the objective function $J(\delta,Y_p^m)$ is defined in equation (\ref{cost}). Then it follows from the KKT conditions that we obtain necessary conditions for optimal solutions to Optimization Problem (\ref{formulation}) as follows \cite{man94}.
\begin{prop}
The optimal solution $\delta^{*}$ to Optimization Problem (\ref{formulation}) with the multipliers $\mu_1$ and $\mu_2$ satisfies the KKT conditions
\begin{equation}\label{kkt}
\begin{split}
&P^m_{e}(\delta^*)^T\left(\frac{\partial P^m_{e}}{\partial\delta}|_{\delta^*}\right)+\mu_1-\mu_2=0 \\
&\delta^*-\bar{\delta}+x^2_1=0 \\
&\delta^*-\underline{\delta}-x^2_2=0 \\
&\mu_1(\delta^*-\bar{\delta})=0 \\
&\mu_2(\delta^*-\underline{\delta})=0 \\
&\mu_1-y_1^2=0 \\
&\mu_2-y_2^2=0 \\
\end{split}
\end{equation}
where $x_i$ and $y_i$, $i\in I_2$ are the unknown variables.
\end{prop}

\begin{proof}
The KKT conditions for Optimization Problem (\ref{formulation}) are composed of four components: stationary, primal feasibility, dual feasibility and complementary slackness. Specifically,
stationary condition allows us to obtain
$$
\frac{\partial J(\delta,Y^m_p)}{\partial\delta}|_{\delta^*}+\mu_1-\mu_2=0,
$$
which is equivalent to
$$
P^m_{e}(\delta^*)^T\left(\frac{\partial P^m_{e}}{\partial\delta}|_{\delta^*}\right)+\mu_1-\mu_2=0
$$
using equation (\ref{cost}). Additionally, the primal feasibility leads to $\underline{\delta}\leq\delta^*\leq\bar{\delta}$, which can be converted into equality constraints
$$
\delta^*-\bar{\delta}+x_1^2=0, \quad \delta^*-\underline{\delta}-x_2^2=0
$$
with the unknown variables $x_1,x_2 \in R$. Moreover, the dual feasibility corresponds to $\mu_1,\mu_2\geq0$, which can be replaced by
$$
\mu_1-y_1^2=0,\quad \mu_2-y_2^2=0
$$
with the unknown variables $y_1,y_2 \in R$. Finally, the complementary slackness gives
$$
\mu_1(\delta^*-\bar{\delta})=0, \quad \mu_2(\delta^*-\underline{\delta})=0
$$
This completes the proof.
\end{proof}

\begin{remark}
To reduce the computation burden, the partial derivative of $P^m_{e}$ with respective to $\delta$ can be approximated by
\begin{equation}\label{app}
\frac{\partial P^m_{e}}{\partial\delta}|_{\delta^*}\approx\frac{P^m_{e}(\delta^*+\epsilon)-P^m_{e}(\delta^*)}{\epsilon}
\end{equation}
with the sufficiently small $\epsilon$.
\end{remark}

\section{Numerical Method} \label{sec:num}
To avoid the computation of partial derivatives and reduce computation costs, the Jacobian Free Newton Krylov (JFNK) Method is employed to solve the system of nonlinear algebraic equations without forming the Jacobian matrix. Essentially, the JFNK methods are synergistic combinations of Newton methods for solving nonlinear equations and Krylov subspace methods for solving linear equations \cite{kno04}.
To facilitate the analysis, we rewrite Equation (\ref{kkt}) in matrix form
\begin{equation}\label{mequation}
    \mathcal{F}(\mathbf{z})=\mathbf{0}
\end{equation}
where $\mathbf{z}=(\delta^*,\mu_1,\mu_2,x_1,x_2,y_1,y_2)^T\in R^7$ denotes the unknown vector, and $\mathbf{0}\in R^7$ refers to a zero vector. To obtain the iterative formula for solving (\ref{mequation}), we compute Taylor series of $\mathcal{F}(\mathbf{z})$ at $\mathbf{z}^{s+1}$ as follows
\begin{equation}\label{taylor}
    \mathcal{F}(\mathbf{z}^{s+1})=\mathcal{F}(\mathbf{z}^{s})+\mathfrak{J}(\mathbf{z}^s)(\mathbf{z}^{s+1}-\mathbf{z}^{s})+O(\delta\mathbf{z}^{s})
\end{equation}
with $\delta\mathbf{z}^{s}=\mathbf{z}^{s+1}-\mathbf{z}^{s}$. By neglecting the high-order term $O(\delta\mathbf{z}^{s})$ and setting $\mathcal{F}(\mathbf{z}^{s+1})=\mathbf{0}$, we obtain
\begin{equation}\label{lequ}
    \mathfrak{J}(\mathbf{z}^s)\cdot\delta\mathbf{z}^{s}=-\mathcal{F}(\mathbf{z}^{s}),\quad  \quad s\in Z^{+}
\end{equation}
where $\mathfrak{J}(\mathbf{z}^s)$ represents the Jacobian matrix and $s$ denotes the iteration index. Thus, solutions to Equation (\ref{mequation}) can be approximated by implementing Newton iterations
$$
\mathbf{z}^{s+1}=\mathbf{z}^{s}+\delta\mathbf{z}^{s}
$$
where $\delta\mathbf{z}^{s}$ is obtained by Krylov methods. First of all, the Krylov subspace is constructed as follows
\begin{equation}\label{kspace}
    K_i=\mathrm{span}\left(\mathbf{r}^s,~\mathfrak{J}(\mathbf{z}^s)\mathbf{r}^s,~\mathfrak{J}(\mathbf{z}^s)^2\mathbf{r}^s,...,~\mathfrak{J}(\mathbf{z}^s)^{i-1}\mathbf{r}^s\right)
\end{equation}
with
$$
\mathbf{r}^s=-\mathcal{F}(\mathbf{z}^{s})-\mathfrak{J}(\mathbf{z}^s)\cdot\delta\mathbf{z}_0^{s},
$$
where $\delta\mathbf{z}_0^{s}$ is the initial guess for the Newton correction and is typically zero \cite{kno04}. Actually, the optimal solution to $\delta\mathbf{z}^{s}$ is the linear combination of elements in Krylov subspace $K_i$.
\begin{equation}\label{gmres}
    \delta \mathbf{z}^s=\delta \mathbf{z}_0^s+\sum_{j=1}^{i-1}\beta_j\cdot \mathfrak{J}(\mathbf{z}^s)^j\mathbf{r}^s
\end{equation}
where $\beta_j$, $j\in\{1,2,...,i-1\}$ is obtained by minimizing $\|\mathfrak{J}(\mathbf{z}^s)\delta z^s+\mathcal{F}(\mathbf{z}^s)\|_2$ with Generalized Minimal RESidual (GMRES) method \cite{saad86}. In particular, matrix-vector products in (\ref{gmres}) can be approximated by
\begin{equation}\label{jv}
\mathfrak{J}(\mathbf{z}^s)\mathbf{r}^s\approx\frac{\mathcal{F}(\mathbf{z}^s+\xi\mathbf{r}^s)-\mathcal{F}(\mathbf{z}^s)}{\xi}
\end{equation}
where $\xi$ is a sufficiently small value \cite{saad90}. In this way, we avoid the computation of Jacobian matrix via matrix-vector products in (\ref{jv}) while solving Equation (\ref{mequation}).

Table~\ref{jfnk} summarizes the JFNK method for solving Equation (\ref{mequation}). First of all, we set the initial step $s=0$, the initial tolerance $\epsilon_0$ and the minimum tolerance $\epsilon_{\min}$ for evaluating the termination condition of loop iterations. Then the residual $\mathbf{r}^s$ is calculated in each iteration, which allows us to construct the Krylov subspace $K_i$. For elements in $K_i$, the matrix-vector products are approximated by Equation (\ref{jv}) without forming the Jacobian. Next, the gradient $\delta\mathbf{z}^{s}$ for Newton iterations is obtained via GMRES method. Finally, we update the tolerance $\epsilon_s$ and step number $s$ after implementing the Newton iteration for $\mathbf{z}^s$. And a new iteration loop is launched if the termination condition $\epsilon_{s}\leq\epsilon_{\min}$ does not hold.
\begin{table}
 \caption{\label{jfnk} JFNK Method.}
 \begin{center}
 \begin{tabular}{lcl} \hline
  1: Set $s=0$, $\epsilon_{0}$ and $\epsilon_{\min}$ satisfying $\epsilon_{0}>\epsilon_{\min}$  \\
  2: \textbf{while}~($\epsilon_{s}>\epsilon_{\min}$) \\
  3: ~~~~~~~Calculate the residual $\mathbf{r}^s=-\mathcal{F}(\mathbf{z}^{s})-\mathfrak{J}(\mathbf{z}^s)\cdot\delta\mathbf{z}_0^{s}$ \\
  4: ~~~~~~~Construct the Krylov subspace $K_i$ in (\ref{kspace}) \\
  5: ~~~~~~~Approximate matrix-vector products in (\ref{gmres}) using (\ref{jv})   \\
  6: ~~~~~~~Compute $\beta_j$ in (\ref{gmres}) with GMRES method \\
  7: ~~~~~~~Compute $\delta\mathbf{z}^{s}$  with (\ref{gmres}) \\
  8: ~~~~~~~Update $\mathbf{z}^{s+1}=\mathbf{z}^{s}+\delta\mathbf{z}^{s}$ \\
  9: ~~~~~~~Update $\epsilon_{s+1}=\|\delta\mathbf{z}^{s}\|/\|\mathbf{z}^s\|$ \\
 10: ~~~~~Update $s=s+1$ \\
 11: \textbf{end while} \\ \hline
 \end{tabular}
 \end{center}
\end{table}

\begin{table}
 \caption{\label{cia} Contingency Identification Algorithm.}
 \begin{center}
 \begin{tabular}{lcl} \hline
  1: Select the disturbed branch \\
  2: Set $l_{\max}$, $l=0$ and $\delta=0$ \\
  3: \textbf{while}~($l<l_{\max}$) \\
  4: ~~~~~~~Compute $\delta^*$ with the JFNK method \\
  5: ~~~~~~~\textbf{if} ($J(\delta^*,B^m)<J(\delta,B^m)$) \\
  6: ~~~~~~~~~~~$\delta=\delta^*$ \\
  7: ~~~~~~~\textbf{end if} \\
  8: ~~~~~~~Update $l=l+1$ \\
  9: \textbf{end while} \\ \hline
 \end{tabular}
 \end{center}
\end{table}

Table \ref{cia} presents the explicit process of implementing Contingency Identification Algorithm (CIA). First of all, we select a branch in power transmission networks to add the disturbance with the initial value $\delta=0$, and the maximum iterative step $l_{\max}$ is specified with the initial iterative step $l=0$. Then we compute the optimal disturbance $\delta^*$ with the JFNK method. The disturbance value $\delta^*$ in (\ref{kkt}) is saved if it leads to the worse cascading blackout ($i.e.$, $J(\delta^*,B^m)<J(\delta,B^m)$). The above algorithm does not terminate until the iterative step $l$ is larger than or equal to $l_{\max}$.

\begin{remark}
Essentially, the proposed approach to identifying initial disturbances is universal, which also applies to power distribution systems using the AC power flow model and more complicated protective mechanisms. The main difference lies in the computation of transmission power at the final cascading step, $i.e.$, $P_e^m$ in Equation (\ref{kkt}).
\end{remark}

\section{Simulation and Validation}\label{sec:sim}

In this section, we implement the proposed CIA in Table \ref{cia} to search for the disruptive disturbances on selected branches of IEEE 118 Bus System \cite{zim11}. The numerical results on disruptive disturbances are validated by disturbing the selected branch with the computed magnitude of disturbance. Per-unit system is adopted with the base value of $100$ MVA in numerical simulations, and the power flow threshold for each branch is specified in Table \ref{thre}. The power flow on each branch is close to the saturation, although it does not exceed their respective threshold. In this way, the power system is vulnerable to initial contingencies, and thus is likely to suffer from cascading blackouts.
\begin{table}
 \caption{\label{thre} Thresholds of the transmission power on each branch.}
 \begin{center}
 \begin{tabular}{lcl} \hline
  \textbf{Power Flow Threshold}~~~~~~~~~~~~~~~~~~~~~~~~~~\textbf{Branch ID} \\ \hline
  ~~~~~~~~~~~7~~~~~~~~~~~~~~~~~~~~~~~~~~~~~  32 \\ \hline
  ~~~~~~~~~~~6~~~~~~~~~~~~~~~~~~~~~~~~~~~~~  18 31 \\ \hline
  ~~~~~~~~~~~5~~~~~~~~~~~~~~~~~~~~~~~~~~~~~  7 8 9 \\ \hline
  ~~~~~~~~~~~4~~~~~~~~~~~~~~~~~~~~~~~~~~~~~  1 12 13 14 21 33 36 37 96 \\ \hline
  ~~~~~~~~~~~3~~~~~~~~~~~~~~~~~~~~~~~~~~~~~  11 15 41 51 141 \\  \hline
  ~~~~~~~~~~~~~~~~~~~~~~~~~~~~~~~~~~~~~~~~~~ 2 3 5 6 10 17 19 20 22 23 25 26 27 28 29 30 \\
  ~~~~~~~~~~~2~~~~~~~~~~~~~~~~~~~~~~~~~~~~~  34 39 42 43 54 62 90 93 94 97 98 99 104 105 \\
  ~~~~~~~~~~~~~~~~~~~~~~~~~~~~~~~~~~~~~~~~~~ 106 107 108 126 127 137 139 163 178 179 183 \\ \hline
  ~~~~~~~~~~~1~~~~~~~~~~~~~~~~~~~~~~~~~~~~~  all other branches \\ \hline
 \end{tabular}
 \end{center}
\end{table}

\begin{figure}\centering
 {\includegraphics[width=0.8\textwidth]{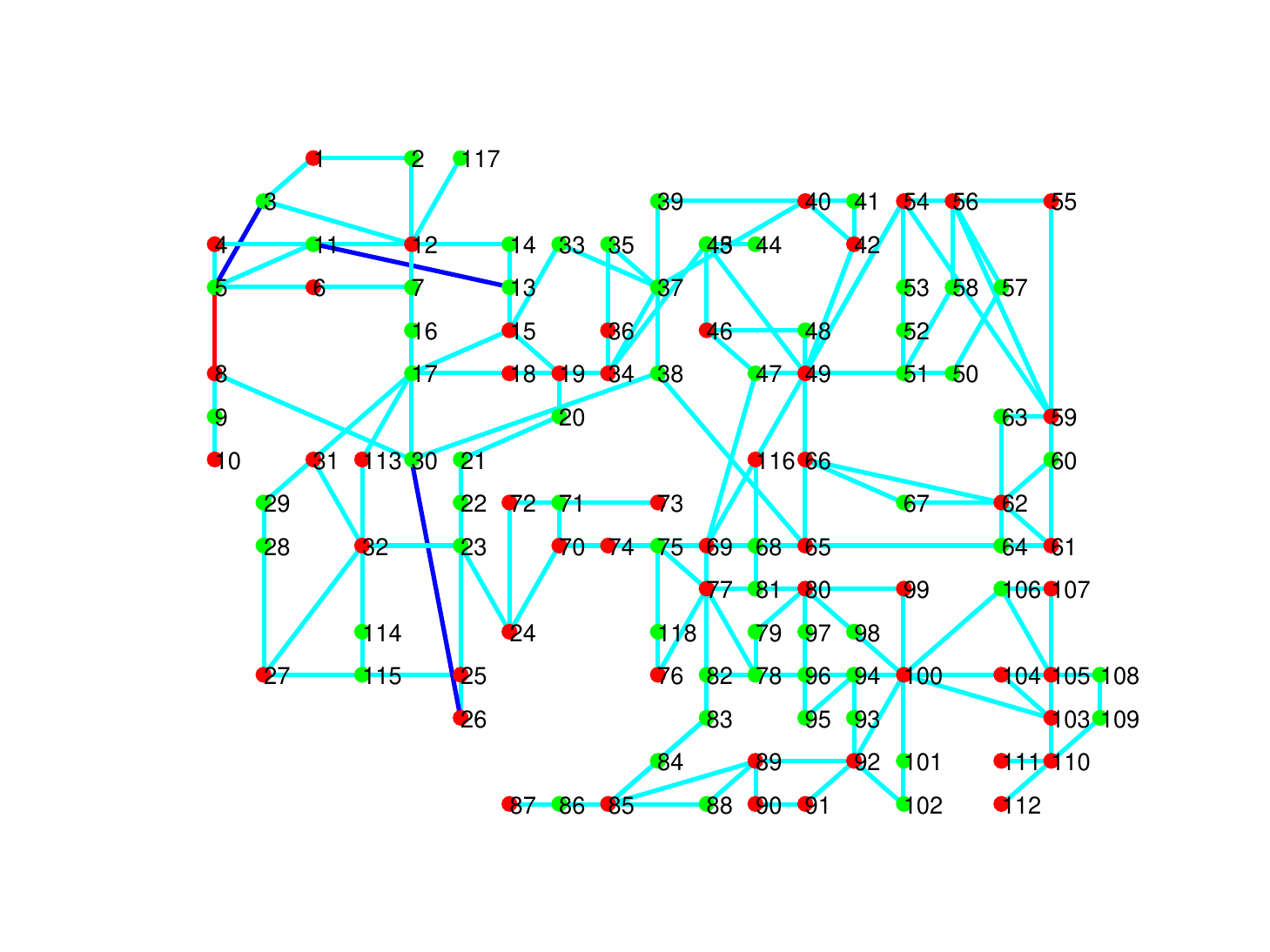}}
 \caption{\label{Step1} Initial state of IEEE 118 Bus System. Red balls denote the generator buses, while blue ones stand for the load buses. Cyan lines represent the branches of power systems. In addition, the red line is selected as the disturbed branch, and three blue lines are the HVDC links, including Branch 4, Branch 16 and Branch 38.}
\end{figure}
Figure \ref{Step1} shows the normal state of IEEE 118 Bus System, which includes 53 generator buses, 64 load buses, 1 reference bus (Bus $69$) and 186 branches. Branch $8$ (red link connecting Bus $5$ to Bus $8$) is randomly selected as the disturbed element of power networks. And the HVDC links are denoted by blue lines including Branch 4 connecting Bus $3$ to Bus $5$, Branch 16 connecting Bus $11$ to Bus $13$ and Branch 38 connecting Bus $26$ to Bus $30$. The maximum iterative step $l_{\max}$ is equal to $10$ in the CIA. Other parameters are given as follows: $\epsilon=10^{-2}$ in Equation (\ref{app}), $\epsilon_{\min}=10^{-8}$ in the JFNK method, $\underline{\delta}=0$, $\bar{\delta}=37.45$ and $m=12$. For simplicity, we specify the same values for the parameters of three HVDC links as follows: $R_{ci}=R_{cr}=R_L=0.1$, $\alpha=\pi/15$ and $\gamma=\pi/4$.
Regarding the FACTS devices, we set $X_{min,i}=0$, $X_{max,i}=10$ and $X_{ref,i}=0$ for the TCSC dynamics, and $K_P=4$, $K_I=3$ and $K_D=2$ for its PID controller. Additionally, the reference transmission power $P^{ref}_{e,i}$ is equal to the threshold of transmission power $\sigma_i$. We consider two preset values of the timer in protective relays, $i.e.$, $T=0.5$s and $T=1$s. Contingency Identification Algorithm is carried out to search for the disturbance that results in the worst-case cascading failures of power systems. For the IEEE 118 Bus System without the FACTS devices, the computed magnitude of disturbance on Branch $8$ is $37.45$, which exactly leads to the outage of Branch 8. For the power system with the FACTS devices and the preset time of the timer $T=0.5$s, the computed disturbance magnitude is $36.77$, while it is $35.98$ for $T=1$s.

\begin{figure}\centering
 {\includegraphics[width=0.8\textwidth]{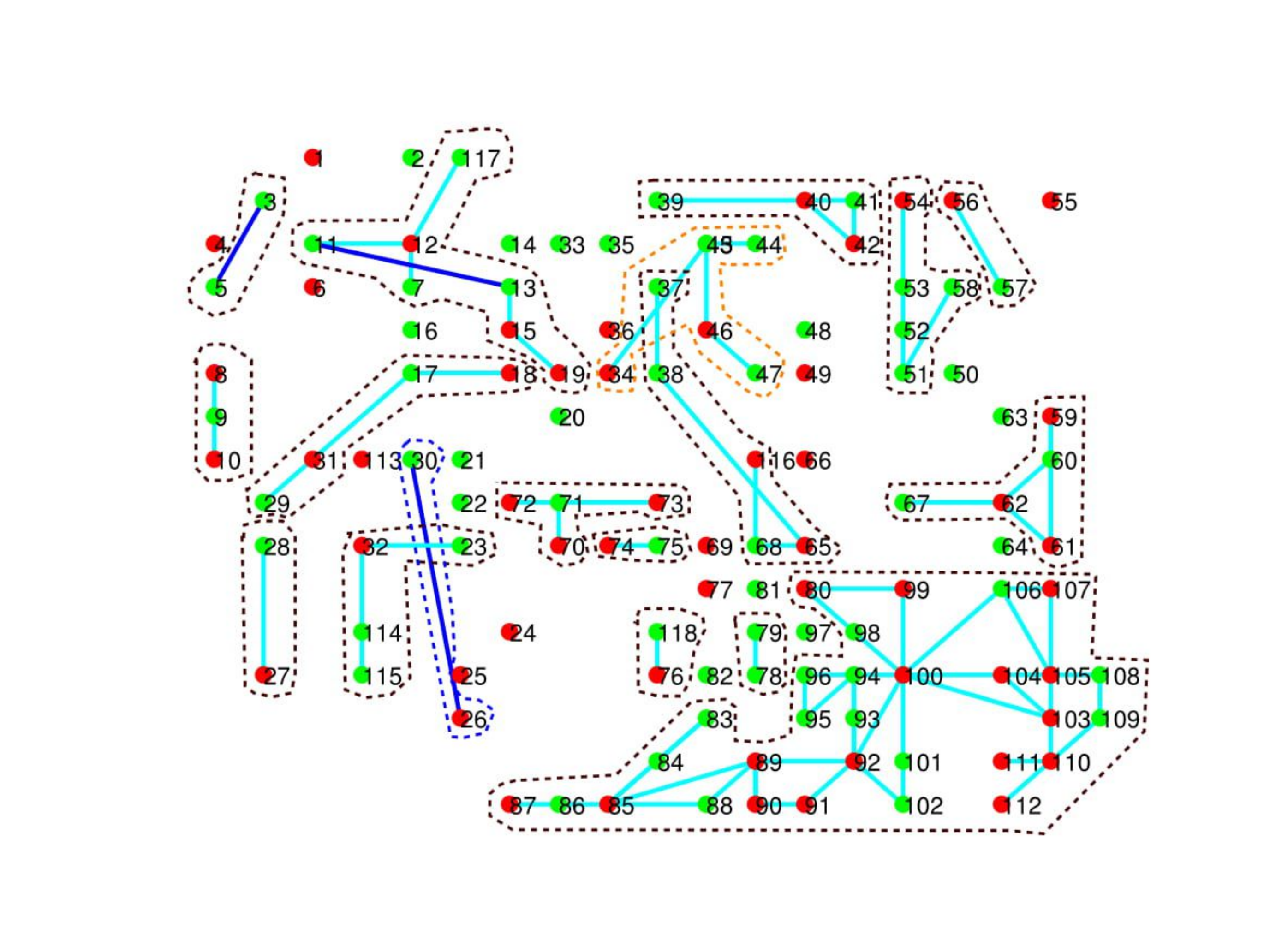}}
 \caption{\label{nofacts} Final configuration of IEEE 118 Bus System without FACTS devices.}
\end{figure}

Next, we validate the proposed approach by adding the computed disturbances on Branch 8 of IEEE 118 Bus Systems. Specifically,
Fig. \ref{nofacts} demonstrates the final state of IEEE 118 Bus System with no FACTS devices and $T=1$s. The cascading process terminates with 95 outage branches and the cost function value of $53.28$ after 16 seconds, and the system collapses with 42 islands in the end. These 42 islands include $24$ isolated buses and $18$ subnetworks encircled by the dashed lines. In contrast, Figure \ref{facts05} presents the final configuration of IEEE 118 Bus Systems with the protection of the FACTS devices and $T=0.5$s. The cascading process ends up with 40 outage branches and the cost function value of $102.56$ after 10 seconds, and the power system is separated into $17$ islands, which include $6$ subnetworks and $11$ isolated buses. Figure \ref{facts1} gives the final state of power systems with the FACTS devices and $T=1$s. It is observed that the power network is eventually split into 3 islands (Bus 14, Bus 16 and a subnetwork composed of all other buses) with only $6$ outage branches and the cost function of $153.69$. Note that the initial disturbances from CIA fail to cause the outage of Branch $8$ in the end for both $T=0.5$s and $T=1$s. The above simulation results demonstrate the advantage of the FACTS devices in preventing the propagation of cascading outages. The larger preset time of timer enables the FACTS devices to sufficiently adjust the branch impedance in response to the overload stress. And thus the less severe damages are caused by the contingency for the larger preset time of timer.

Figure \ref{comp} presents the time evolution of outage branches in IEEE 118 Bus System as a result of disturbing Branch $8$ in three different scenarios. The cyan squares denote the number of outage branches with no FACTS devices and $T=1$s, while the green and blue ones refer to the numbers of outage branches with the FACTS devices for $T=0.5$s and $T=1$s, respectively. The computed disturbances are added to change the admittance of Branch $8$ at $t=0$s. With no FACTS devices, the cascading outage of branches propagates quickly from $t=2$s to $t=10$s and terminates at $t=16$s. When the FACTS devices are adopted and the preset time of timer is $T=0.5$s, the cascading failure starts at $t=2$s and speeds up till $t=8$s and stops at $t=10$s. For $T=1$s, the cascading outage propagates slowly due to the larger preset time of timer and comes to an end with only $6$ outage branches at $t=8$s. Together with protective relays and HVDC links, the FACTS devices succeed in protecting power systems against blackouts by adjusting the branch impedance in real time. More precisely, the number of outage branches decreases by $57.9\%$ with FACTS devices and $T=0.5$s and decreases by $93.7\%$ with FACTS devices and $T=1$s.

\begin{figure}\centering
{\includegraphics[width=0.8\textwidth]{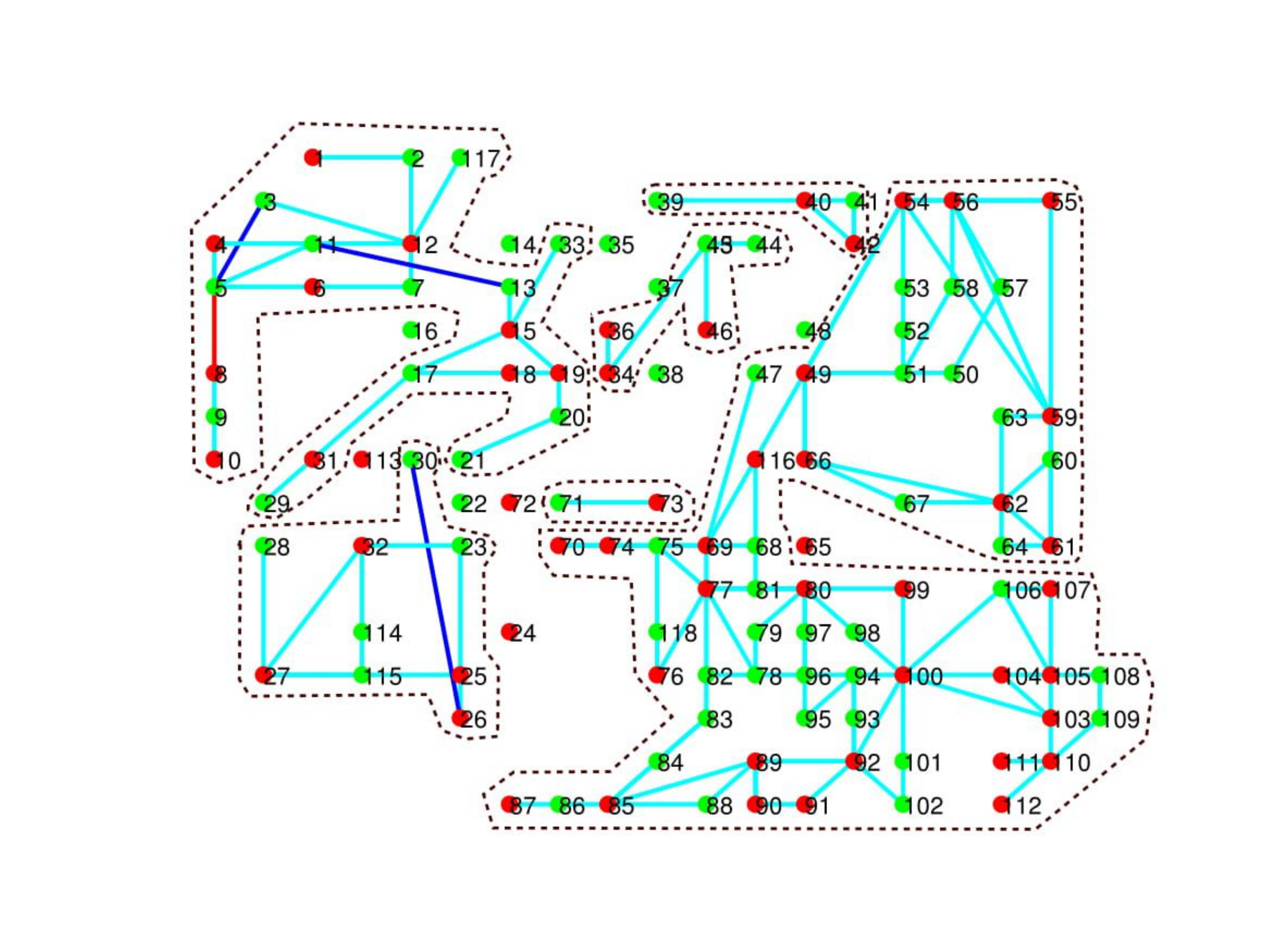}}
\caption{\label{facts05} Final configuration of IEEE 118 Bus System with FACTS devices and $T=0.5$s.}
\end{figure}

\begin{figure}\centering
{\includegraphics[width=0.8\textwidth]{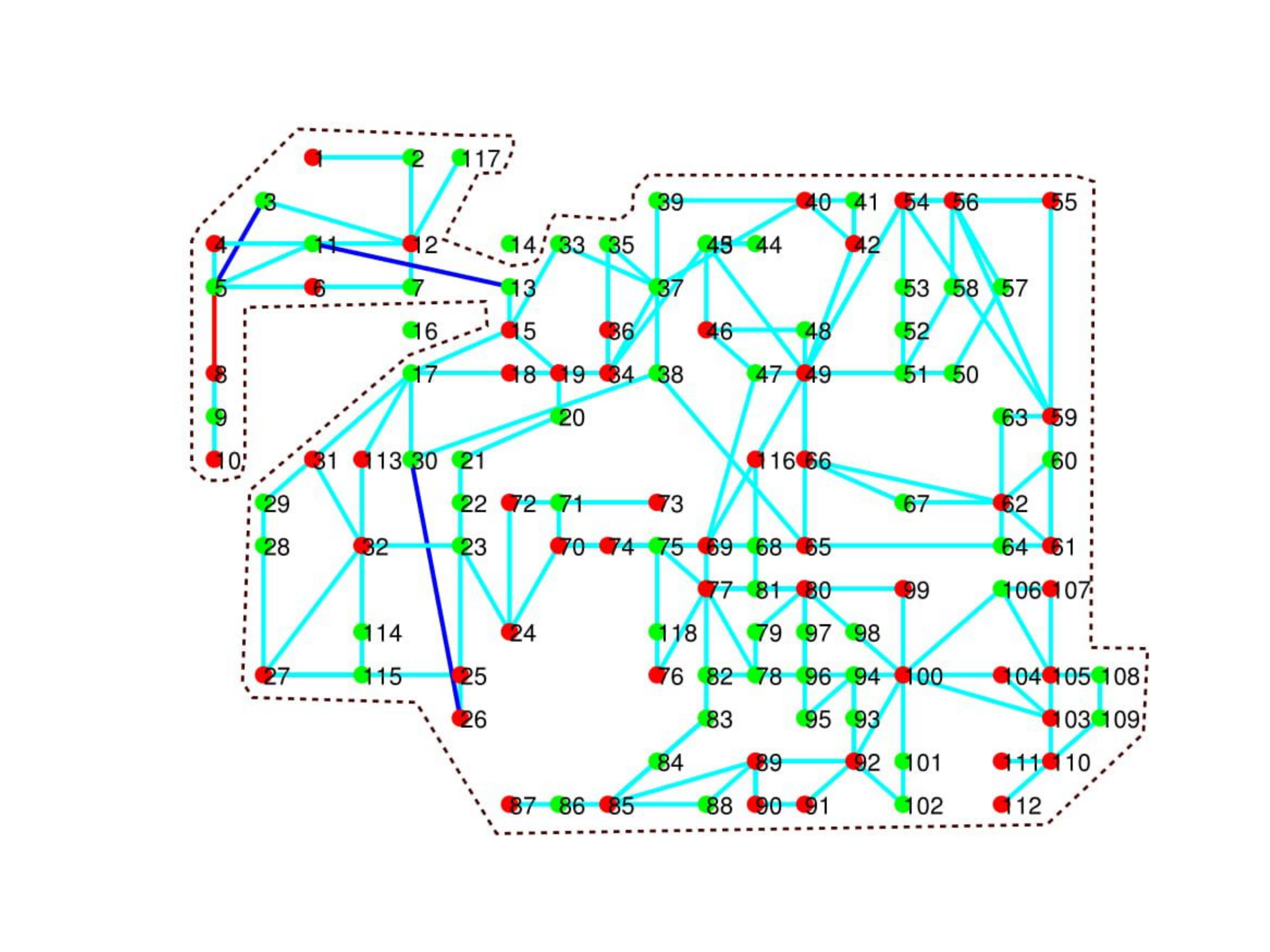}}
\caption{\label{facts1} Final configuration of IEEE 118 Bus System with FACTS devices and $T=1$s.}
\end{figure}

\begin{figure}\centering
{\includegraphics[width=0.6\textwidth]{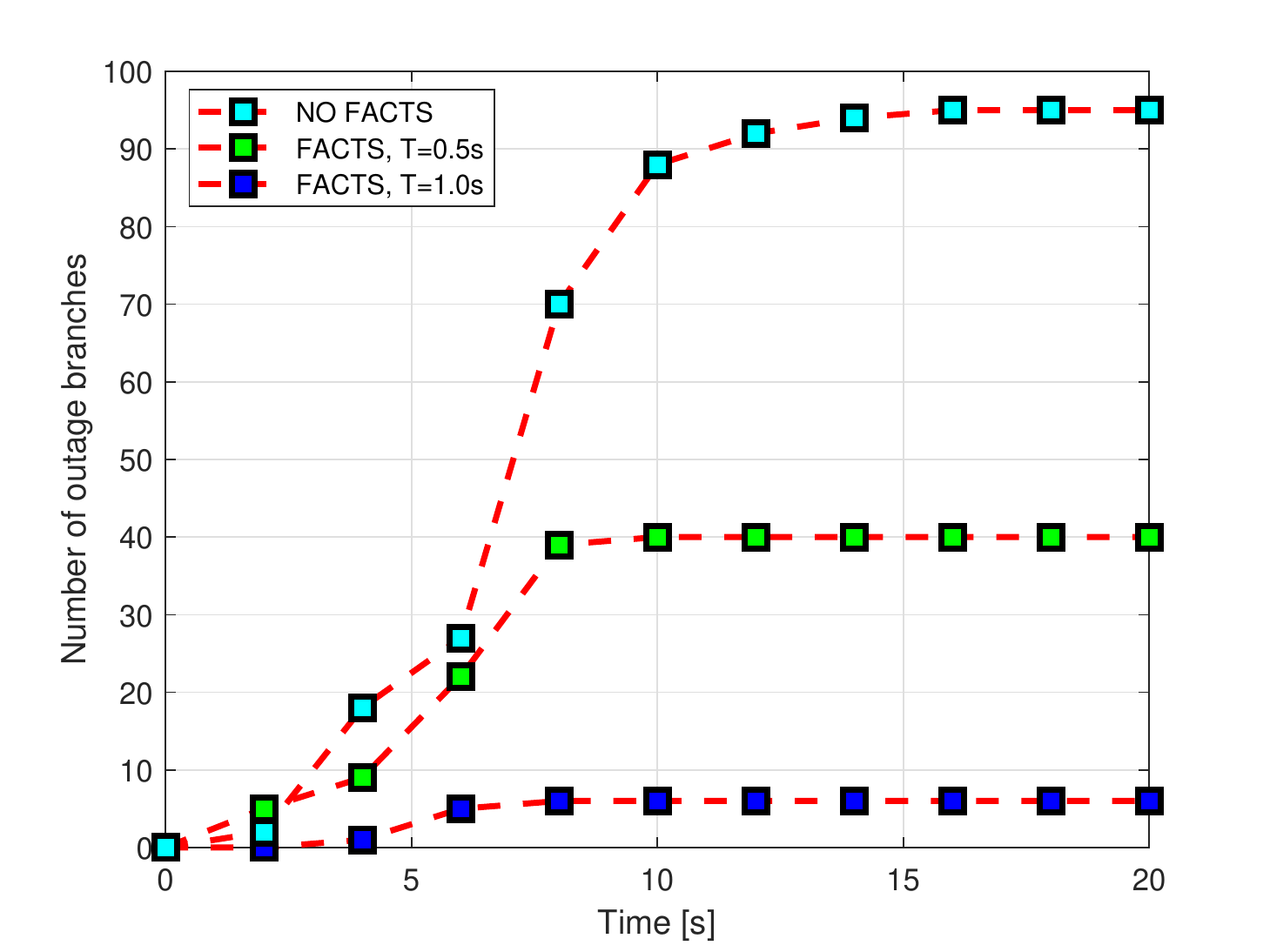}}
\caption{\label{comp} Time evolution of outage branches during cascading blackouts.}
\end{figure}

\section{Conclusions}\label{sec:con}

In this paper, we investigated the problem of identifying the initial contingencies that lead to cascading blackout of power transmission networks equipped with FACTS devices, HVDC links and protective relays. An optimization formulation was proposed to identify the contingencies in the framework of nonlinear programming, and an efficient numerical method was presented
to solve the optimization problem. Numerical simulations were carried out on IEEE 118 Bus Systems to validate the proposed approach. Significantly, the contingency identification algorithm allows us to detect some nontrivial disturbances that lead to the severe cascading failure of power transmission networks, other than severing the branch. It is demonstrated that the coordination of FACTS devices and protective relays greatly enhances the capability of power grids against blackouts. In the next step, we will investigate the contingency identification problem for the AC-based power distribution systems with transient process.

\section*{Acknowledgment}

This work is partially supported by the Future Resilience System Project at the Singapore-ETH Centre (SEC), which is funded by the National Research Foundation of Singapore (NRF) under its Campus for Research Excellence and Technological Enterprise (CREATE) program. It is also supported by Ministry of Education of Singapore under Contract MOE2016-T2-1-119.

\end{document}